\documentclass[10pt,conference]{IEEEtran}
\usepackage{epsfig,graphics,subfigure,psfrag,amsmath,cases}
\usepackage{latexsym,amssymb,amsmath,epsfig,subfigure,algorithm}
\usepackage{algorithmic}
\usepackage{color}
\usepackage{url}
\usepackage{scrtime}

\author{\authorblockN{ Derrick Wing Kwan Ng and Robert Schober}
Institute for Digital Communications, Universit\"at Erlangen-N\"urnberg, Germany\\
Email: kwan@lnt.de,  schober@lnt.de

}

\title{Resource Allocation for Secure Communication in Systems with Wireless Information and Power Transfer }

\date{\thistime,\,\today}

\newtheorem{proposition}{Proposition}

\DeclareMathOperator{\Tr}{Tr}
\DeclareMathOperator{\Rank}{Rank}

\newtheorem{Remark}{Remark}

\newcommand{\abs}[1]{\lvert#1\rvert}
\newcommand{\norm}[1]{\lVert#1\rVert}
\begin{document}

\maketitle

\begin{abstract}
This paper considers secure communication in a multiuser multiple-input single-output (MISO) downlink system with
simultaneous wireless information and power transfer. We study the design of resource
allocation algorithms minimizing the total transmit power for the case when
the receivers  are able to harvest energy from the radio frequency. In particular,
the algorithm design is formulated as a non-convex
optimization problem which takes into account artificial
noise generation to combat potential eavesdroppers, a minimum required signal-to-interference-plus-noise ratio (SINR) at the desired receiver, maximum tolerable SINRs at the potential eavesdroppers,  and a minimum required power delivered to the receivers.
We adopt a semidefinite programming (SDP) relaxation  approach to obtain an upper bound solution for the considered problem. The tightness of the upper bound is revealed by examining a sufficient condition for the global optimal solution. Inspired by the sufficient condition,  we propose two suboptimal resource allocation schemes enhancing secure communication and facilitating efficient energy harvesting. Simulation results demonstrate a close-to-optimal performance achieved by the proposed suboptimal schemes and significant transmit power savings by  optimization of the artificial noise generation.

\end{abstract}

\renewcommand{\baselinestretch}{0.97}
\large\normalsize

\section{Introduction}
\label{sect1}
 Green radio communications has received much attention in both industry and academia under the pressure of environmental concerns and the rapidly increasing cost of
energy \cite{JR:Mag_green}-\nocite{CN:WIPT_fundamental,CN:Shannon_meets_tesla,JR:WIPT_fullpaper}\cite{CN:WIP_receiver}. In particular, different resource allocation algorithms and the use of multiple antennas  have been proposed  in the literature
 for energy saving in wireless communication systems.  Unfortunately, portable devices are often powered by batteries with limited operating cycles which remain the bottlenecks in
perpetuating the lifetime of networks. The introduction of energy harvesting capabilities for communication devices is considered as a vital solution in providing  self-sustainability  to power limited communication systems \cite{CN:WIPT_fundamental}-\cite{CN:WIP_receiver}.  In addition to harvesting energy from a variety of natural renewable energy sources such as wind and solar, exploiting ambient background electromagnetic radiation in radio frequency (RF) is also a viable source of energy for energy scavenging. More importantly, wireless energy harvesting technology facilitates the possibility of simultaneous wireless information and power transfer which introduces a paradigm shift in system and resource allocation algorithm design  \cite{CN:WIPT_fundamental}-\cite{CN:WIP_receiver}.   In practice, the transmitter can increase the energy of the
information carrying signal for facilitating energy harvesting at the receivers. However, this may also increases the
susceptibility to eavesdropping due a higher potential for information leakage.


On the other hand, recently,  a large amount of work has been
devoted to information-theoretic physical (PHY) layer
security \cite{Report:Wire_tap}-\nocite{JR:Artifical_Noise1}\cite{JR:Kwan_physical_layer}, as an alternative or complement to cryptographic encryption. In fact, PHY layer security exploits the physical characteristics  of the wireless  fading channel for providing perfect secrecy of communication. In  \cite{Report:Wire_tap}, Wyner showed
 that when the source-eavesdropper channel is a degraded version of the source-destination channel,   the source and the destination can exchange perfectly secure
 messages at a non-zero rate. As a result, secure communication systems employing multiple antennas have been proposed. By exploiting the extra degrees of freedom
offered by multiple antennas, artificial noise is
injected into the null space of the channels of the desired receiver  to degrade the
 channels of potential eavesdroppers.
  In \cite{JR:Artifical_Noise1} and
 \cite{JR:Kwan_physical_layer}, the authors
proposed different power allocation algorithms for maximizing the ergodic secrecy capacity and outage secrecy capacity via artificial noise generation, respectively. However, the receivers in  \cite{JR:Artifical_Noise1} and
 \cite{JR:Kwan_physical_layer} are assumed to be powered by perpetual energy sources which may not always be possible for power limited systems.  Furthermore, a significant portion of transmit power is allocated to artificial noise generation \cite{JR:Artifical_Noise1,JR:Kwan_physical_layer}  to guarantee secure communication. Indeed, the artificial noise can be used as an energy harvesting source for the receivers in extending the lifetime of the network. Yet, the proposed algorithms in \cite{JR:Artifical_Noise1,JR:Kwan_physical_layer}  do not utilize  the artificial noise for energy harvesting. Besides, the works in \cite{CN:WIPT_fundamental}-\cite{Report:Wire_tap}
focus on single antenna transmitters and their results may not be able to provide quality of services in secure communication systems with energy harvesting receivers.

Motivated by the aforementioned observations, we
formulate the resource allocation algorithm design for secure communication in multiuser multiple-input single-output (MISO) systems with concurrent wireless information and power transfer
 as an optimization problem. A semidefinite programming (SDP) based resource allocation algorithm is proposed to obtain an upper bound solution. Subsequently, the upper bound solution is used as a building block for the design of two suboptimal schemes which provide close-to-optimal performance.

 \begin{figure*}
 \centering
\includegraphics[width=5in]{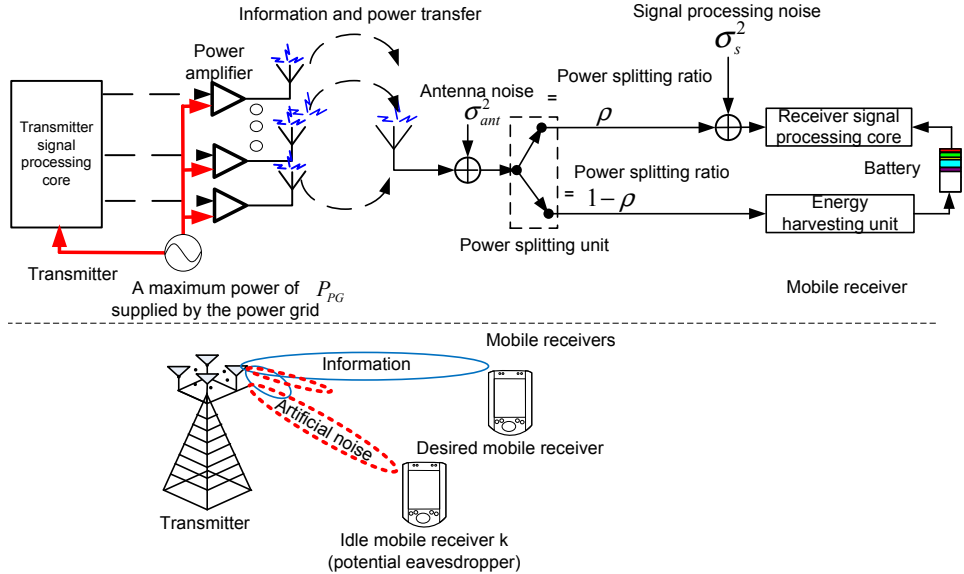}
 \caption{Multiuser system model for $K=2$ mobile receivers.  The upper half of the figure illustrates the block diagram of the transceiver
model for wireless information and power transfer.} \label{fig:system_model}
\end{figure*}

\section{System Model}
\label{sect:OFDMA_AF_network_model}
In this section, after introducing the notation used in this paper, we present the adopted multiuser downlink channel model.
\subsection{Notation}
 For square-matrix $\mathbf{X}$,
$\Tr(\mathbf{X})$ denotes the  trace of matrix
$\mathbf{X}$ and $\mathbf{X}\succeq 0$ indicates that
$\mathbf{X}$ is a positive semidefinite matrix. $(\mathbf{X})^H$
and $\Rank(\mathbf{X})$ denote the conjugate transpose and the
rank of matrix $\mathbf{X}$, respectively. Matrix $\mathbf{I}_{N}$
denotes an $N\times N$ identity matrix.  $\mathbb{C}^{N\times M}$ denotes the space of $N\times M$ matrices with complex entries.
$\mathbb{H}^N$ represents the set of all $N$-by-$N$ complex Hermitian matrices.
The distribution of a circularly symmetric complex Gaussian (CSCG)
vector with mean vector $\mathbf{x}$ and covariance matrix
$\mathbf{\Sigma}$  is denoted by ${\cal
CN}(\mathbf{x},\mathbf{\Sigma})$, and $\sim$ means ``distributed
as".   $[x]^+=\max\{0,x\}$ and $\cal E\{\cdot\}$ denotes the statistical expectation.

\subsection{ Channel Model}
We consider the downlink of a communication system which consists of a transmitter and $K$ legitimate receivers. The transmitter is equipped with $N_t$ transmit antennas while the receivers are single antenna devices and are able to decode information and harvest energy from radio signals, cf.  Figure \ref{fig:system_model}. In each scheduling slot, the transmitter conveys information to a given receiver and transfers energy\footnote{In this paper,  a normalized energy unit, i.e., Joule-per-second, is adopted.  Therefore,
the terms "power" and "energy" are used interchangeably in this paper.} to all receivers. The $K-1$ idle receivers are legitimate receivers and supposed to harvest power from RF when they are inactive. However, it is possible that the idle receivers are malicious and eavesdrop the information signal of other legitimate receivers. As a result, they are potential eavesdroppers, which should be taken into account for providing secure communication.  The total
 bandwidth of the system is $\cal B$ Hertz.
 We assume a frequency flat slow fading channel and we focus on a time division duplexing (TDD) system. The
downlink channel gains of all receivers can be accurately obtained based on the uplink pilots in the handshaking signal via channel reciprocity. The downlink
received signal at the desired receiver and the $K-1$ idle receivers are given by, respectively,
\begin{eqnarray}
y&=&\mathbf{h}^{H} \mathbf{x}+ z_s +z_a,\\
y_{I,k}&=&\mathbf{g}_{k}^{H} \mathbf{x}+ z_s +z_a,\,\,  \forall k=\{1,\ldots,K-1\},
\end{eqnarray}
where $\mathbf{x}\in\mathbb{C}^{ N_T \times 1}$ denotes the transmitted symbol vector.
$\mathbf{h}^{H}\in\mathbb{C}^{1\times N_T}$ is the channel
vector between the transmitter and the desired receiver
  and
$\mathbf{g}_{k}^{H} \in\mathbb{C}^{1\times N_T}$ is the channel
vector between the transmitter and idle receiver (potential eavesdropper) $k$. We note that both
variables, $\mathbf{h}$ and $\mathbf{g}_{k}$,
  include the effects of the multipath fading and path loss
of the associated channels.
 $z_s$ and  $z_a$ are  additive white Gaussian noises (AWGN) resulting from signal processing and
  the receive antenna with zero means and variances $\sigma_{s}^2$ and $\sigma_{ant}^2$, respectively\footnote{We assume that the signal processing and thermal noise characteristics are the same for all receivers due to a similar hardware architecture.},  cf. Figure \ref{fig:system_model}.

  \subsection{Hybrid Information and Energy Harvesting Receiver}
\label{sect:receiver}
In this paper, we adopt hybrid receivers \cite{JR:WIPT_fullpaper,CN:WIP_receiver}
which can split the received signal into two power streams with power splitting ratio $1-\rho$ and $\rho$, cf. Figure \ref{fig:system_model}, for harvesting energy and decoding the modulated  information in the signal, respectively. The power splitting unit is assumed to be a perfect passive device; it does not introduce any extra power gain, i.e., $0\le \rho\le 1$, or noise in power splitting.  Besides, we assume that there is a battery for storing the harvested energy for future use. In practice, if the amount of harvested energy is lager than what can currently be stored, the excess harvested energy will be discarded. We assume that the harvested energy  is used by receivers as a supplementary energy source for supporting their normal operation.

\subsection{Artificial Noise Generation}
In order to provide
secure communication to the desired receiver, artificial noise
signals are generated at the transmitter to degrade the channels between
the transmitter and the idle receivers (potential eavesdroppers). The transmitter chooses the transmit signal vector $\mathbf{x}$ as
\begin{eqnarray}
\mathbf{x}=\underbrace{\mathbf{w}s}_{\mbox{Desired signal}}+\underbrace{\mathbf{v}}_{\mbox{Artificial noise}},
\end{eqnarray}
where $s\in\mathbb{C}^{1\times 1}$ and $\mathbf{w}\in\mathbb{C}^{N_t\times 1}$  are the information bearing signal and the corresponding  beamforming vector dedicated to the desired receiver, respectively. We assume without loss of generally that  ${\cal E}\{\abs{s}^2\}=1$. $\mathbf{v}\in\mathbb{C}^{N_t\times 1}$ is the artificial noise vector generated by the transmitter to combat the potential eavesdroppers. $\mathbf{v}$ is modeled as a complex Gaussian random vector with
\begin{eqnarray}
\mathbf{v}\sim {\cal CN}(\mathbf{0}, \mathbf{V}),
\end{eqnarray}
where $\mathbf{V}\in \mathbb{H}^{N_t}, \mathbf{V}\succeq \mathbf{0}$, denotes the covariance matrix of the artificial noise.

We note that unlike in other system models used in the literature, e.g. \cite{JR:Artifical_Noise1,JR:Kwan_physical_layer}, the artificial noise signals in the considered system can act as a vital energy source which supplies energy to the receivers. The transmitter can use the energy of the information signal solely as a energy supply for the receivers \cite{JR:WIPT_fullpaper,CN:WIP_receiver}. However, increasing the transmit power of the information signal for facilitating energy harvesting may also increases the susceptibility to  eavesdropping. As a result, in this paper, we advocate the dual use of artificial noise in providing simultaneous security and efficient energy harvesting.

\section{Resource Allocation Algorithm Design}\label{sect:forumlation}
\subsection{System Capacity and Secrecy Capacity}
\label{subsect:Instaneous_Mutual_information}
 Given perfect channel state information (CSI) at the
receiver, the system capacity (bit/s/Hz) between the transmitter and the desired receiver
is given by
\begin{eqnarray}\label{eqn:cap}
C&=&\log_2\Big(1+\Gamma\Big)\,\,\,\,
\mbox{and}\,\,\\
\Gamma&=&\frac{\rho\abs{\mathbf{h}^H\mathbf{w}}^2}
{\rho(\sigma_{ant}^2+\Tr(\mathbf{h}\mathbf{h}^H\mathbf{V}))+\sigma_s^2} ,
\end{eqnarray}
where $\Gamma$ is the received signal-to-interference-plus-noise ratio (SINR) at the desired receiver.

On the other hand, the channel capacity between the transmitter and idle receiver (potential eavesdropper) $k$  is given  by
\begin{eqnarray}\label{eqn:cap-eavesdropper}
C_{I,k}&=&\log_2\Big(1+\Gamma_{k}\Big)\,\,\,\,
\mbox{and}\,\,\\
\label{eqn:eavesdropper-SINR}
\Gamma_{I,k}&=&\frac{\rho_k\abs{\mathbf{g}_k^H\mathbf{w}}^2}{\rho_k(\sigma_{ant}^2+\Tr(\mathbf{g}_k\mathbf{g}_k^H\mathbf{V}))+\sigma_s^2}  \\ \label{eqn:eavesdropper-SINR-bound}
&\stackrel{(a)}{\le}& \frac{\abs{\mathbf{g}_k^H\mathbf{w}}^2}{\sigma_{ant}^2+\Tr(\mathbf{g}_k\mathbf{g}_k^H\mathbf{V})+\sigma_s^2}   \end{eqnarray}
where $\rho_k$ and  $\Gamma_{I,k}$ are the power splitting ratio and the received SINR at idle receiver $k$, respectively. (a) is due to the fact that $\Gamma_{I,k}$ is a monotonically increasing function of $\rho_k$. The physical meaning of (\ref{eqn:eavesdropper-SINR-bound}) is that idle receiver $k$ gives up the opportunity to harvest energy  and devotes all the received power to eavesdropping.  Therefore, the maximum achievable secrecy capacity between the transmitter
and the desired receiver can be expressed as \cite{JR:Artifical_Noise1}
\begin{eqnarray}\label{eqn:secrecy_cap}
C_{sec}=\Big[C - \underset{k\in\{1,\ldots,K-1\}}{\max} C_{I,k}^{UP}\Big]^+,
\end{eqnarray}
where $C_{I,k}^{UP}$ is obtained by replacing SINR $\Gamma_{I,k}$ in (\ref{eqn:eavesdropper-SINR}) with its upper bound in (\ref{eqn:eavesdropper-SINR-bound}).

\subsection{Optimization Problem Formulation}
\label{sect:cross-Layer_formulation}
The optimal resource allocation policy, ${\mathbf w}^*$, ${\rho}^*$ ,${\mathbf  V}^*$, for minimizing the total radiated power,  can be
obtained by solving
\begin{eqnarray}
\label{eqn:cross-layer}&&\hspace*{5mm} \min_{\mathbf{V}\in \mathbb{H}^{N_t},\mathbf{w}, \rho
}\,\, \norm{\mathbf{w}}^2+\Tr(\mathbf{V})\nonumber\\
\notag \mbox{s.t.} &&\hspace*{-5mm}\mbox{C1: }\notag\frac{\rho\abs{\mathbf{h}^H\mathbf{w}}^2}{\rho(\sigma_{ant}^2+\Tr(\mathbf{h}\mathbf{h}^H
\mathbf{V}))+\sigma_s^2} \ge \Gamma_{req}, \\
&&\hspace*{-5mm}\mbox{C2: }\notag\frac{\abs{\mathbf{g}_k^H\mathbf{w}}^2}
{\sigma_{ant}^2+\Tr(\mathbf{g}_k\mathbf{g}_k^H\mathbf{V})+\sigma_s^2} \le \Gamma_{tol_k},\forall k,\\
&&\hspace*{-5mm}\mbox{C3: }\notag(1-\rho)\eta\abs{\mathbf{h}^H\mathbf{w}}^2+(1-\rho)\eta\Big(\Tr(\mathbf{h}\mathbf{h}^H\mathbf{V})+\sigma_{ant}^2\Big)\\
&&\ge P_{\min}, \notag\\
&&\hspace*{-5mm}\mbox{C4: }\notag \eta\abs{\mathbf{g}_k^H\mathbf{w}}^2+\eta\Big(\Tr(\mathbf{g}_k\mathbf{g}_k^H\mathbf{V})+\sigma_{ant}^2\Big)\ge P_{\min_k},\forall k, \\
&&\hspace*{-5mm}\mbox{C5: }\notag \norm{\mathbf{w}}^2  +\Tr(\mathbf{V})\le P_{\max}, \\
&&\hspace*{-5mm}\mbox{C6: }\notag \norm{\mathbf{w}}^2\varepsilon  +\Tr(\mathbf{V})\varepsilon+P_C\le P_{PG}, \\
&&\hspace*{-5mm}\mbox{C7:}\,\, 0\le\rho\le 1,\quad\mbox{C8:}\,\, \mathbf{V}\succeq 0.
\end{eqnarray}
Variable $\Gamma_{req}$ in C1 specifies the minimum requirement on the SINR
 of the desired receiver for information decoding. $\Gamma_{tol_k}$ in C2 denotes the maximum tolerable SINR at idle receiver (potential eavesdropper) $k$. In practice, the transmitter sets $\Gamma_{req}\gg \Gamma_{{tol_k}},\forall k\in\{1,\ldots,K-1\}$, to ensure secure communication. Specifically, if the above optimization problem is feasible, it is guaranteed that the secrecy capacity $C_{sec}\ge \log_2(1+\Gamma_{req})-\log_2(1+\underset{k}{\max}\{\Gamma_{{tol_k}}\})\ge 0$. We note that although $\Gamma_{req}$ and $\Gamma_{{tol_k}}$ in C1 and C2, respectively,  are
not optimization variables in this paper, a balance between
secrecy capacity and system capacity can be struck
by varying their values. $P_{\min}$ and $P_{\min_k}$ in C3 and C4 set the minimum required power transfer to the desired information  receiver and potential eavesdroppers, respectively.
We note that the transmitter can only guarantee the minimum required power transfer to the idle receivers if they employ all their received power for energy harvesting, i.e., if they do not intend to eavesdrop.
$\eta$ denotes the energy harvesting efficiency of the receivers in converting
the received radio signal to electrical energy for storage. $P_{\max}$ in C5 restricts the maximum transmit spectrum mask for reducing the amount of out-of-cell
interference and the value is specified by regulation.
Constants
$P_C$ and $\varepsilon$ in C6 account for the circuit
power consumption at the transmitter and the inefficiency of the
power amplifier, respectively. C6 is imposed to guarantee that the total power consumption of the transmitter for both transmission and circuitries is less than the maximum  power supplied by the power grid $P_{PG}$, cf. Figure \ref{fig:system_model}. C7 is the boundary constraint for power splitting variable $\rho$. C8 and $\mathbf{V}\in \mathbb{H}^{N_t}$ constrain matrix $\mathbf{V}$ to be a  positive semidefinite Hermitian matrix to satisfy the physical requirements on covariance matrices.

\section{Solution of the Optimization Problem} \label{sect:solution}
The optimization problem in (\ref{eqn:cross-layer}) can be classified as a non-convex quadratically constrained quadratic program (QCQP). The non-convexity is due to
constraints C1 and C3 on the information bearing beamforming vector $\mathbf{w}$ and the power splitting ratio $\rho$.
In general,
there is no standard approach for solving non-convex optimization problems. In some extreme cases, a brute force approach is required to obtain a global optimal solution which is computationally intractable for a moderate  system size.
In order to derive an efficient
resource allocation algorithm for the considered problem, we recast the problem as a
convex optimization problem by semidefinite programming (SDP) relaxation. In the sequel, we assume that the problem is always feasible for studying the design of different resource allocation schemes.
\subsection{Semidefinite Programming Relaxation} \label{sect:solution_dual_decomposition}
For facilitating the SDP relaxation, we define $\mathbf{W}=\mathbf{w}\mathbf{w}^H$ and rewrite problem (\ref{eqn:cross-layer}) in terms of $\mathbf{W}$ as
\begin{eqnarray}
\label{eqn:SDP}&&\hspace*{5mm} \min_{\mathbf{W,V}\in \mathbb{H}^{N_t}, \rho
}\,\, \Tr(\mathbf{W})+\Tr(\mathbf{V})\nonumber\\
\notag \mbox{s.t.} &&\hspace*{-5mm}\mbox{C1: }\notag\frac{\rho\Tr(\mathbf{h}\mathbf{h}^H\mathbf{W})}
{\rho(\sigma_{ant}^2+\Tr(\mathbf{h}\mathbf{h}^H\mathbf{V}))+\sigma_s^2} \ge \Gamma_{req}, \\
&&\hspace*{-5mm}\mbox{C2: }\notag\frac{\Tr(\mathbf{g}_k\mathbf{g}_k^H\mathbf{W})}{\sigma_{ant}^2+\Tr(\mathbf{g}_k\mathbf{g}_k^H\mathbf{V})+\sigma_s^2} \le \Gamma_{{tol_k}},\forall k, \\
&&\hspace*{-5mm}\mbox{C3: }\notag \Tr(\mathbf{hh}^H\mathbf{W})+\Tr(\mathbf{h}\mathbf{h}^H\mathbf{V})+\sigma_{ant}^2\ge\frac{ P_{\min}}{(1-\rho)\eta}, \\
&&\hspace*{-5mm}\mbox{C4: }\notag\Tr(\mathbf{g}_k\mathbf{g}_k^H\mathbf{W})+\Tr(\mathbf{g}_k\mathbf{g}_k^H\mathbf{V})+\sigma_{ant}^2\ge \frac{P_{\min_k}}{\eta},\forall k, \\
&&\hspace*{-5mm}\mbox{C5: }\notag \Tr(\mathbf{W})  +\Tr(\mathbf{V})\le P_{\max}, \\
&&\hspace*{-5mm}\mbox{C6: }\notag \Tr(\mathbf{W})\varepsilon  +\Tr(\mathbf{V})\varepsilon+P_C\le P_{PG}, \\
&&\hspace*{-5mm}\mbox{C7:}\,\, 0\le\rho\le 1, \,\,\,\mbox{C8:}\,\, \mathbf{W}\succeq 0, \mathbf{V}\succeq 0,\notag \\
&&\hspace*{-5mm}\mbox{C9:}\,\, \Rank(\mathbf{W})=1,
\end{eqnarray}
where $\mathbf{W}\succeq 0$, $\mathbf{W}\in \mathbb{H}^{N_t}$, and $\Rank(\mathbf{W})=1$ in (\ref{eqn:SDP}) are imposed to guarantee that $\mathbf{W}=\mathbf{w}\mathbf{w}^H$. By relaxing constraint $\mbox{C9: }\Rank(\mathbf{W})=1$, i.e., removing it from the problem formulation, the considered problem becomes a convex SDP which can be solved efficiently by numerical solvers such as SeDuMi \cite{JR:SeDumi}. From the basic principles of optimization theory, if the obtained solution $\mathbf{W}$ for the relaxed problem is a rank one matrix, then it is the optimal solution of the original problem in (\ref{eqn:SDP}). However, it is known that the relaxation may not be tight and in that case the result of the relaxed problem serves as a performance upper bound for the original problem. In the following, we will reveal a sufficient condition for $\Rank(\mathbf{W})=1$ of the relaxed problem and exploit it as a
building block for the design of two suboptimal resource allocation schemes.

 \subsection{Optimality Conditions for SDP Relaxation }
In this subsection, we reveal the tightness of the proposed  SDP relaxation via examination of the dual problem and the Karush-Kuhn-Tucker (KKT) conditions of the relaxed version of problem (\ref{eqn:SDP}). For this purpose, we first need
the Lagrangian function  of  (\ref{eqn:SDP}) which is given by
\begin{eqnarray}\hspace*{-2mm}&&\notag{\cal
L}(\mathbf{W},\mathbf{V},\rho,\lambda,\boldsymbol{\beta},\mu,\boldsymbol{\delta},\theta,\psi,\mathbf{Y},\mathbf{Z})\\
\notag\hspace*{-5mm}&=&\hspace*{-3mm} \Tr(\mathbf{W})+\Tr(\mathbf{V})-\mathbf{YW}-\mathbf{ZV} \\
\notag\hspace*{-5mm}&+&\hspace*{-3mm} \sum_{k=1}^{K-1}\beta_k \Big[\hspace*{-0.5mm}
\Tr(\mathbf{g}_k\mathbf{g}_k^H\mathbf{W})-
\Gamma_{tol_{k}}\hspace*{-0.5mm}\Big(
\Tr(\mathbf{g}_k\mathbf{g}_k^H\mathbf{V}) \hspace*{-0.5mm}+\hspace*{-0.5mm}\sigma_{ant}^2\hspace*{-0.5mm}+\hspace*{-0.5mm}
\sigma_s^2\Big)\hspace*{-0.5mm}\Big]\\
\notag\hspace*{-5mm}&+&\hspace*{-3mm}  \mu\Big(\frac{ P_{\min}}{\eta(1-\rho)}\Tr(\mathbf{hh}^H\mathbf{W})-\Tr(\mathbf{hh}^H\mathbf{V})-\sigma_{ant}^2\Big)\end{eqnarray}\begin{eqnarray}
\notag\hspace*{-5mm}&+&\hspace*{-3mm} \sum_{k=1}^{K-1} \delta_k \Big( \frac{P_{\min_k}}{\eta}-\Tr(\mathbf{g}_k\mathbf{g}_k^H\mathbf{W})-\Tr(\mathbf{g}_k\mathbf{g}_k^H\mathbf{V})-
\sigma_{ant}^2\Big)\\
\notag\hspace*{-5mm}&+& \hspace*{-3mm}\notag \theta\Big( \Tr(\mathbf{W})\varepsilon  +\Tr(\mathbf{V})\varepsilon+P_C- P_{PG}\Big)  \\
\notag\hspace*{-5mm}&+& \hspace*{-3mm}\notag \psi\Big( \Tr(\mathbf{W}) +\Tr(\mathbf{V})- P_{\max}\Big)\\
\hspace*{-5mm}&+&\hspace*{-3mm} \lambda\Big(\frac{\Gamma_{req}\sigma_s^2}{\rho}-\Tr(\mathbf{h}\mathbf{h}^H\mathbf{W})+\Gamma_{req}
\Tr(\mathbf{h}\mathbf{h}^H\mathbf{V})+\sigma_{ant}^2\hspace*{-0.5mm}\Big).
\label{eqn:Lagrangian}
\end{eqnarray}
Here, $\lambda\ge 0$ is the Lagrange multiplier for the minimum required SINR of the desired receiver in C1. $\boldsymbol \beta$ is
the vector of Lagrange multipliers for the maximum tolerable SINRs of the potential eavesdroppers in C2 with elements ${\beta_k}\ge 0$,
$k\in\{1,\,\ldots,\,K-1\}$.  Lagrange multiplier  $\mu\ge 0$  corresponds to the minimum required power transfer to the desired receiver in C3. $\boldsymbol \delta$,  with elements $\delta_{k}\ge 0$,  is the Lagrange multiplier vector
associated with the minimum required power transfer to the potential eavesdroppers in C4.
 $\psi,\theta\ge0$ are the Lagrange multipliers for the maximum radiated power and the total power usage from
the power grid in C5 and C6, respectively. Matrices $\mathbf{Y,Z}\succeq 0$ are the Lagrange multipliers for the semidefinite constraints on matrices $\mathbf{W}$ and $\mathbf{V}$ in C8, respectively.  On the other hand, boundary constraint C7 for $\rho$
is  satisfied automatically as will be illustrated
when we study  the sufficient condition for the optimal resource allocation solution in the
Appendix. Thus, the dual problem for the SDP relaxed problem is given by
\begin{eqnarray}\label{eqn:dual}
\underset{ \underset{\mathbf{Y},\mathbf{Z}\succeq 0}{\lambda,\boldsymbol{\beta},\mu,\boldsymbol{\delta},\theta,\psi\ge0}}{\max}\ \underset{{\mathbf{W,V}\in \mathbb{H}^{N_t},\rho}}{\min}{\cal
L}(\mathbf{W},\mathbf{V},\rho,\lambda,\boldsymbol{\beta},\mu,\boldsymbol{\delta},\theta,\psi,\mathbf{Y},\mathbf{Z}).\label{eqn:master_problem}
\end{eqnarray}

Now, we are ready to reveal a sufficient condition for a rank one matrix solution for the relaxed version of problem (\ref{eqn:SDP}) in the following proposition.
\begin{proposition} Consider the relaxed version  of problem (\ref{eqn:SDP}) for $\Gamma_{req}>0$ and suppose that the problem is feasible. Then, $\Rank(\mathbf{W})=1$ when  $\beta_k\ge\delta_k\ge 0 $, $\forall k$.
\end{proposition}
\begin{proof}Please refer to Appendix.
\end{proof}
Intuitively, when the requirement in constraint C4 becomes less stringent, i.e., $P_{\min_k}\rightarrow 0 \Rightarrow \delta_k\rightarrow 0$, the  SDP relaxation algorithm has a higher chance to obtain a rank one matrix solution and thus achieves the global optimal.

\begin{Remark}
We would like to emphasize that although Proposition 1 provides a sufficient condition for a rank one matrix solution under SDP relaxation, we found by simulation that there are instances in which SDP relaxation results in a rank one matrix even though the sufficient condition does not hold.
\end{Remark}

In the following, we propose two suboptimal
resource allocation schemes which are inspired by the
SDP relaxation based resource allocation solution.

\subsubsection{Suboptimal Resource Allocation Scheme 1}
It can be observed that the solution of SDP relaxation has rank one $\mathbf{W}$ when constraint C4 is not active, i.e., $\delta_k=0,\forall k$, or it is independent of optimization variable $\mathbf{W}$. For facilitating an efficient resource allocation scheme design, we replace constraint  C4 in (\ref{eqn:SDP}) by C10 and the new optimization is given as follows:
\begin{eqnarray}\label{eqn:suboptimal1}
&&\hspace*{5mm} \min_{\mathbf{W,V}\in \mathbb{H}^{N_t}, \rho
}\,\, \Tr(\mathbf{W})+\Tr(\mathbf{V})\\
\notag \mbox{s.t.} &&\hspace*{10mm}\mbox{C1, C2, C3, C5, C6, C7, C8, }\\
&&\hspace*{-5mm}\mbox{C10: } \Tr(\mathbf{g}_k\mathbf{g}_k^H\mathbf{V})+\sigma_{ant}^2\ge \frac{P_{\min_k}}{\eta}, \forall k\in\{1,\ldots,K-1\}\notag .
\end{eqnarray}

Compared to constraint C4, the new constraint C10 does not take into account the contribution of the information signal to the harvested power at the potential eavesdroppers, as the term $\Tr(\mathbf{g}_k\mathbf{g}_k^H\mathbf{W})$ is neglected.  Since replacing constraint C4 by C10 results in a smaller feasible solution set for the original problem, the obtained solution of problem (\ref{eqn:suboptimal1}) serves as a performance lower bound for the original optimization problem (\ref{eqn:cross-layer}). We note that the new constraint does not destroy the convexity of the relaxed problem and the relaxed problem can be solved efficiently via SDP relaxation and the numerical methods  suggested in Section \ref{sect:solution_dual_decomposition} \cite{book:convex}. In fact, it can be shown that the obtained solution\footnote{We can follow a similar approach as in the Appendix to examine the KKT conditions for the new problem formulation. In particular,  the sufficient conditions for a rank one matrix solution stated in Proposition 1 are always satisfied for the new problem formulation since  constraint C10 is independent of $\mathbf{W}$.} of problem (\ref{eqn:suboptimal1}) has always rank one, i.e., $\Rank(\mathbf{W})=1$, even though SDP relaxation is applied.

\subsubsection{Suboptimal Resource Allocation Scheme 2}
The proposed suboptimal scheme 2 is a hybrid resource allocation scheme which is summarized in Table \ref{table:algorithm}. Specifically, it computes the solutions for  the SDP relaxation in (\ref{eqn:SDP}) and suboptimal scheme 1 in parallel and selects one of the solutions. In particular, when the solution for the SDP relaxation is not rank one, i.e., $\Rank(\mathbf{W})>1$; the upper bound solution of SDP relaxation is not tight, and thus,  the proposed scheme 2 will adopt the solution given by the proposed suboptimal scheme 1. Otherwise,  the proposed scheme 2 will select the solution given by the SDP relaxation since the global optimal is achieved if $\Rank(\mathbf{W})=1$. We note that although the proposed scheme 2 requires solving two optimization problems, it is still a scheme with polynomial time complexity due to the convexity of both problems.
\begin{table}[t]\caption{Suboptimal Resource Allocation Scheme.}\label{table:algorithm}
\vspace*{-5mm}\small
\begin{algorithm} [H]           \setcounter{algorithm}{1}          
\floatname{algorithm}{Suboptimal Resource Allocation Scheme}
\caption{}          
\label{alg1}                           
\begin{algorithmic} [1]
\normalsize           

\STATE Solve the relaxed version of  problem (\ref{eqn:SDP}) and problem (\ref{eqn:suboptimal1}) in parallel
\IF {the solution of the relaxed version of problem (\ref{eqn:SDP}) is rank one, i.e., $\Rank(\mathbf{W})=1$, } \STATE  $\mbox{Global optimal soultion}=\,$\TRUE \RETURN
 ${\mathbf W}^*$, ${\rho}^*$ ,${\mathbf  V}^*=$ solution of the relaxed version of problem (\ref{eqn:SDP})
 \ELSE
\STATE  $\mbox{Lower bound soultion}=\,$ \TRUE\RETURN
 ${\mathbf W}$, ${\rho}$ ,${\mathbf  V}=$ solution of problem (\ref{eqn:suboptimal1})
 \ENDIF

\end{algorithmic}
\end{algorithm}\vspace*{-5mm}
\end{table}

\section{Results}
\label{sect:result-discussion} In this section, we evaluate the
system performance for the proposed resource allocation schemes using simulations.  The TGn path loss model \cite{report:tgn} for indoor communications is adopted with directional transmit and receive antenna gains of  10 dB. The reference distance of the path loss model is 2 meters  and there are $K$ receivers uniformly distributed between
the reference distance and the maximum service distance of 10 meters.  The system bandwidth is ${\cal B}=200$ kHz. We assume  a carrier
center frequency of $470$ MHz which will be used by the IEEE 802.11af next generation Wi-Fi systems  \cite{report:80211af}. The transmitter is equipped with $N_t=6$ antennas. The small scale fading coefficients
are generated as independent and identically distributed Rician random
variables with Rician factor $6$ dB.  We assume that the signal processing noise in every
receiver is due to thermal noise and quantization noise. Specifically, a 8-bit uniform quantizer is used
for quantizing the received information. As
a result, the quantization noise power and the thermal noise power are $-23$ dBm
and $-111$ dBm, respectively. In addition, the antenna noise is set to $\sigma_{ant}^2= -114$ dBm. Unless specified otherwise, we assume a circuit power
consumption at the transmitter of $P_C=$ 30 dBm, a maximum power supply of $P_{PG}=$ 40 dBm from the power grid, a minimum required power transfer of $P_{\min}=P_{\min_k}=0$ dBm, $\forall k$, and an energy harvesting efficiency of $\eta=0.5$.  On the other hand, we assume the maximum  SINR tolerance of each idle receiver (potential eavesdropper) is $\Gamma_{tol_k}=-10$ dB, $\forall k$, and a power amplifier with power efficiency of
$38\%$ is used at the transmitter,
i.e., $\varepsilon=\frac{1}{0.38}$. The average total transmit power of the transmitter is obtained by averaging over both path loss and multipath  fading.

\begin{figure}[t]
 \centering\vspace*{-0.5cm}
\includegraphics[width=3.5 in]{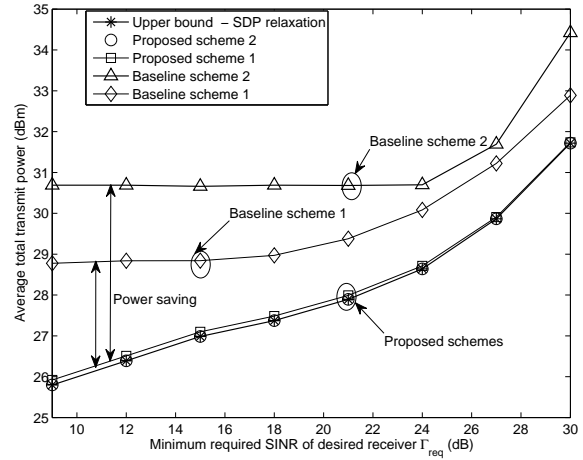}\vspace*{-0.1cm}
\caption{Average total transmit power (dBm) versus the minimum required SINR, $\Gamma_{req}$, of the desired receiver for different resource allocation schemes.
The double-sided arrows indicate the power savings achieved by the proposed schemes compared to the baseline schemes.} \label{fig:p_SNR}\vspace*{-0.5cm}
\end{figure}

\subsection{Average Total Transmit Power and Secrecy Capacity }
Figure \ref{fig:p_SNR} depicts the  average total transmit power versus the
minimum required SINR of the desired receiver, $\Gamma_{req}$, for $K=4$ receivers and different resource allocation schemes.
It can be observed that the average total transmit power of the proposed schemes is a monotonically non-decreasing function of $\Gamma_{req}$. This is attributed to the fact that a higher transmit power is required for satisfying constraint C1 when the requirement of $\Gamma_{req}$ becomes more stringent. Besides, the two proposed suboptimal schemes perform closely to the upper bound system performance achieved by SDP relaxation. In particular, as expected, proposed scheme 1 is less power efficient than  proposed scheme 2 and the upper bound performance. This is because proposed scheme 1 focuses on a smaller feasible solution set and thus obtains a lower bound solution for the original problem formulation (\ref{eqn:cross-layer}). On the other hand, proposed scheme 2 exploits the possibility of achieving the global optimal solution via SDP relaxation and the  lower bound solution. As a result,  it has a superior performance compared to proposed scheme 1.

For comparison, Figure
\ref{fig:p_SNR} also contains the average total transmit power of two baseline
power allocation schemes. For baseline scheme 1, we inject the artificial noise into the null space of the desired receiver such that the artificial noise does not interfere with the desired receiver. Then,  we minimize\footnote{The performances of the two baseline schemes are obtained by solving the corresponding optimization problems via SDP relaxation. Although the solutions may not be rank one,  they serve as performance upper bounds for the two baseline schemes. } the total transmit power by optimizing  $\mathbf{W},\mathbf{V}$, and $\rho$. Baseline scheme 2 has the same structure as baseline scheme 1 excepts that the power splitting ratio is fixed to $\rho=0.5$.
 \begin{figure}[t]
\centering\vspace*{-0.5cm}
\includegraphics[width=3.5 in]{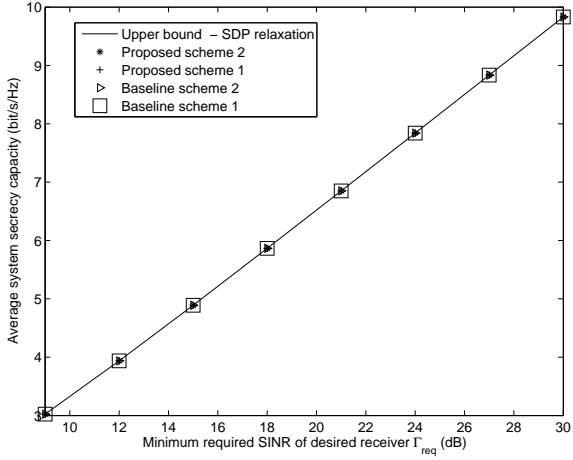}\vspace*{-0.2cm}
 \caption{Average system secrecy capacity (bit/s/Hz) versus the minimum required SINR $\Gamma_{req}$ of the desired receiver for different resource allocation schemes.} \label{fig:cap_SNR}\vspace*{-0.5cm}
\end{figure}
 It can be observed
that the two baseline schemes transmit with higher power than the two proposed suboptimal schemes. Indeed, the proposed suboptimal schemes fully utilize the CSI of all communication links and  optimize the space spanned by artificial noise for performing power allocation. On the contrary, the baseline schemes inject the artificial noise into the null space of the desired receiver. Although the artificial noise does not create any interference to the desired receiver, it is less effective in jamming the potential eavesdroppers. Thus, compared to the proposed schemes, the baseline schemes transmit with higher power on average to satisfy the maximum tolerable SINR constraints of the potential eavesdroppers in C3.

On  the other hand, Figure \ref{fig:cap_SNR} illustrates the average system secrecy capacity versus the minimum required SINR of the desired receiver, $\Gamma_{req}$, for $K=4$ receivers and different resource allocation schemes.  It can be seen that the average system secrecy capacity, i.e., $C_{sec}\ge \log_2(1+\Gamma_{req})-\log_2(1+\underset{k}{\max}\{\Gamma_{{tol_k}}\})$, increases with  $\Gamma_{req}$
since the maximum tolerable SINRs of the idle receivers are constrained to be $\Gamma_{tol_k}=-10$ dB. Besides, all  considered schemes achieve the same value of secrecy capacity. However, the proposed schemes consume much less power than the baseline schemes to achieve the same secrecy capacity, cf. Figure \ref{fig:p_SNR}.

\subsection{Average Total Harvested Power}
Figure \ref{fig:harvested_PT} shows the average total harvested power versus the
minimum required SINR of the desired receiver, $\Gamma_{req}$, for $K=4$ receivers and different resource allocation schemes. It can be seen that the total average harvested power increases with $\Gamma_{req}$. On the one hand, the transmitter has to allocate more power to the information bearing signal to satisfy a larger $\Gamma_{req}$. On the other hand, the power of the artificial noise vector may also increase to reduce the received SINRs of the potential eavesdroppers. As a result, more power is available in the RF and can be harvested by the receivers. Besides, it can be observed that the differences between different schemes in the total energy harvested by the receivers  become smaller when $\Gamma_{req}\gg 1$. Yet, the transmitter of the proposed schemes radiates less power compared to the baseline schemes, cf. Figure \ref{fig:p_SNR},
 due to the proposed optimization.

\begin{figure}[t]
 \centering\vspace*{-0.5cm}
\includegraphics[width=3.5 in]{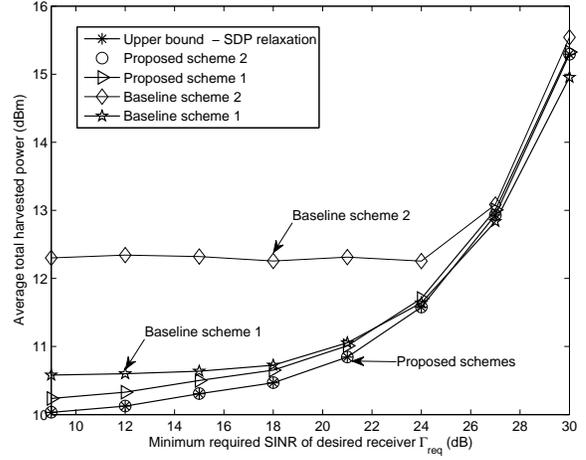}\vspace*{-0.2cm}
\caption{Average total harvested power (dBm) versus
minimum required SINR of the desired receiver, $\Gamma_{req} $ (dB), for different resource allocation schemes. } \label{fig:harvested_PT}\vspace*{-0.4cm}
\end{figure}
 \begin{figure}[t]
\centering
\includegraphics[width=3.5 in]{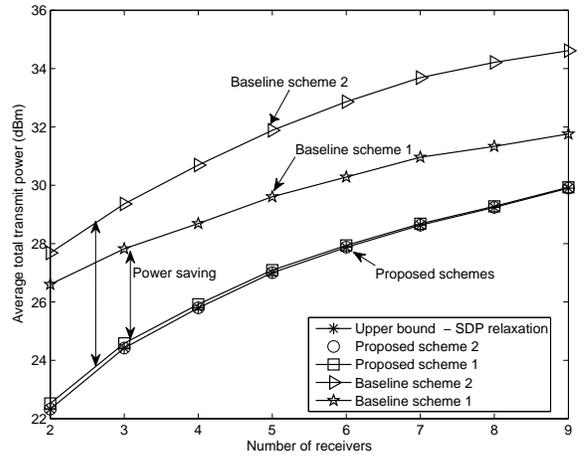}\vspace*{-0.2cm}
 \caption{Average total transmit power (dBm) versus
total number of receivers, $K$, for different resource allocation schemes.} \label{fig:pt_users}\vspace*{-0.5cm}
\end{figure}

\subsection{Average Total Transmit Power versus Number of Receivers}

Figure \ref{fig:pt_users} shows the average total transmit power versus the
number of receivers for different resource allocation schemes. The minimum required SINR of the desired receiver is set to $\Gamma_{req}=9$ dB. It is expected that the total transmit power increases with the number of receivers. The
reason behind this is twofold. First, as the number of receivers in the system increases, there are more idle receivers requiring power transfer from the transmitter even though some of them are experiencing bad channel conditions. Second, there are more potential eavesdroppers present in the system and thus the transmitter has to generate a higher amount of artificial noise to guarantee secrecy.
On the other hand,   the proposed schemes  provide  substantial power savings compared to the  two baseline schemes. We note that, although the corresponding figure is not shown in the paper, the proposed schemes are able to guarantee secure communication even if the number of potential eavesdroppers is larger than the number of transmit antennas $N_t$.

\section{Conclusions}\label{sect:conclusion}
In this paper,  we formulated the resource allocation
algorithm design for secure MISO communication  systems with RF energy harvesting receivers as a non-convex optimization problem.  The problem formulation took into account  secure communication and  power transfer to receivers  via artificial noise injection.   An efficient SDP based resource allocation algorithm was proposed to obtain an upper bound solution for minimization of the total transmit power. The upper bound solution was exploited for the design of two suboptimal but practical resource allocation schemes. Simulation results showed
the excellent performance of the two proposed suboptimal schemes. Furthermore, our results also unveiled the power savings enabled by the optimization of artificial noise generation and the dual use of artificial noise for facilitating   simultaneously security and efficient energy harvesting.

  Studying the impact of imperfect CSI in  secure communication systems with energy transfer is any interesting topic for future work.

\section*{Appendix- Proof of Proposition 1}
It can be shown that the relaxed version of problem (\ref{eqn:SDP}) is convex with respect to the optimization variables and satisfies Slater's constraint qualification. As a result, the KKT conditions are necessary and sufficient conditions \cite{book:convex} for the solution of the relaxed problem. In the following, we focus on the KKT conditions related to $\mathbf{W}$:
\begin{eqnarray}\label{eqn:KKT}
\mathbf{Y}^*\hspace*{-3mm}&\succeq&\hspace*{-3mm} \mathbf{0},\quad\beta_k^*,\,\delta_k^*,\,\lambda^*,\,\psi^*,\,\mu^*\,\ge 0,\\
 \mathbf{Y^*W^*}\hspace*{-3mm}&=&\hspace*{-3mm}\mathbf{0}, \label{eqn:KKT-complementarity}\\ \notag
\label{eqn:KKT_Y2}
\mathbf{Y^*}\hspace*{-3mm}&=&\hspace*{-3mm}\mathbf{I}_{N_t}(1+\psi^*)\hspace*{-0.5mm} +\hspace*{-0.5mm}  \sum_{k=1}^{K-1} \mathbf{g}_k\mathbf{g}^H_k (\beta_k^*-\delta_k^*)     -(\lambda^*+\mu^*)\mathbf{h}\mathbf{h}^H\\
&=&\hspace*{-3mm}\mathbf{A} -(\lambda^*+\mu^*)\mathbf{h}\mathbf{h}^H,
\end{eqnarray}
where $\mathbf{A}=\mathbf{I}_{N_t}(1+\psi^*) + \sum_{k=1}^{K-1} \mathbf{g}_k\mathbf{g}^H_k (\beta_k^*-\delta_k^*)$ and $\mathbf{Y}^*,\,\beta_k^*,\,\delta_k^*,\,\lambda^*,\,\psi^*,\,\mu^*$ are the optimal Lagrange multipliers for (\ref{eqn:dual}). Equation (\ref{eqn:KKT-complementarity}) is the complementary slackness condition and is satisfied when the columns of $\mathbf{W}^*$ lay in the null space of $\mathbf{Y}^*$. Therefore, if  $\Rank(\mathbf{Y}^*)=N_t-1$, then the optimal $\mathbf{W}^*\ne 0$ must be a rank one matrix and the optimal $\mathbf{w}^*$ can be obtained by performing eigenvalue decomposition on $\mathbf{W}^*$.

Now, we prove by contradiction that  $\mathbf{A}$ is a full rank matrix with rank $N_t$ whenever $\beta_k^*\ge\delta_k^*$. Suppose $\mathbf{A}$ is a rank deficient matrix with at least one zero eigenvalue and we denote the associated eigenvector as $\mathbf{u}$. Without loss of generality, we create a matrix  $\mathbf{U}=\mathbf{u}^H\mathbf{u}$ from the eigenvector.  By multiplying both sides of (\ref{eqn:KKT_Y2}) with $\mathbf{U}$ and applying the trace operator, we obtain
\begin{eqnarray}\label{eqn:trace_proof}
\Tr(\mathbf{Y^*}\mathbf{U})&=&\Tr(\mathbf{A}\mathbf{U})   -(\lambda^*+\mu^*)\Tr(\mathbf{h}\mathbf{h}^H\mathbf{U})\notag\\
&\stackrel{(b)}{=}&-(\lambda^*+\mu^*)\Tr(\mathbf{h}\mathbf{h}^H\mathbf{U})
\end{eqnarray}
where ($b$) is due to the fact that $\mathbf{u}$ is generated from the null space of   $\mathbf{A}$ and $\Tr(\mathbf{A}\mathbf{U})=\Tr(\mathbf{u}^H\mathbf{A}\mathbf{u})=0$.
Then, we examine the signs of both sides of the equality in (\ref{eqn:trace_proof}). Consider the right hand side of (\ref{eqn:trace_proof}). We first show $(\lambda^*+\mu^*)>0$. The KKT condition of the optimal power splitting ratio $\rho^*$ leads to
\begin{eqnarray}
\rho^*=\frac{\sqrt{\lambda^*\sigma_s^2 \Gamma_{req}}}{\sqrt{\lambda^*\sigma_s^2 \Gamma_{req}}+\sqrt{\frac{\mu^* P_{\min}}{\eta}}},
\end{eqnarray}
where constraint $\mbox{C7: } 0\le \rho^*\le 1$ is  automatically satisfied. In other words, to achieve a positive minimum required SINR of the desired receiver in C1, $\rho^*>0$ is required which implies $\lambda^*>0$. Thus, $(\lambda^*+\mu^*)>0$ since $\mu^*\ge0$. Then, we prove $\Tr(\mathbf{h}\mathbf{h}^H\mathbf{U})>0$. We note that $\mathbf{g}_k,\forall k,$ belongs to the subspace of $\mathbf{A}$ when $\beta_k^*\ge\delta_k^*$. Also,  $\mathbf{h}$ and $\mathbf{g}_k$ are statistically independent. As a consequence, the probability that $\mathbf{h}$ and $\mathbf{g}_k,\forall k,$ share the same null space is zero which yields $\Tr(\mathbf{h}\mathbf{h}^H\mathbf{U})=\mathbf{u}^H\mathbf{h}\mathbf{h}^H\mathbf{u}\ne 0$. Besides, $\mathbf{h}\mathbf{h}^H$ is a positive semidefinite matrix. Thus, $\Tr(\mathbf{h}\mathbf{h}^H\mathbf{U})$ must be positive and the right hand side of (\ref{eqn:trace_proof}) must be negative.

Next, we focus on the left hand side of (\ref{eqn:trace_proof}).  $\mathbf{Y}^*$ is a positive semidefinite matrix and the left hand side of (\ref{eqn:trace_proof}) is non-negative, i.e.,  $\Tr(\mathbf{Y^*U^*})\ge 0$,  which contradicts the sign of the right hand side of (\ref{eqn:trace_proof}). Therefore, matrix $\mathbf{A}$ must be a full rank matrix with rank $N_t$. From (\ref{eqn:KKT_Y2}), we have
\begin{eqnarray}\notag
&&\Rank(\mathbf{Y}^*)+\Rank((\lambda^*+\mu^*)\mathbf{h}\mathbf{h}^H)\\
 \notag&\ge &\Rank(\mathbf{Y}+(\lambda^*+\mu^*)\mathbf{h}\mathbf{h}^H)\\
&= &\Rank(\mathbf{A})=N_t
\Rightarrow \Rank(\mathbf{Y}^*)\ge N_t-1.
\end{eqnarray}
In other words, $\Rank(\mathbf{Y}^*)$ is either $N_t$ or $N_t-1$. Furthermore, $\mathbf{W}^*\ne\mathbf{0}$ is required to satisfy the minimum SINR requirement of the desired receiver  in C1 when $\Gamma_{req}>0$. As a result, $\Rank(\mathbf{Y}^*)=N_t-1$ and $\Rank(\mathbf{W}^*)=1$.

\bibliographystyle{IEEEtran}
\bibliography{OFDMA-AF}

\begin{thebibliography}{10}
\providecommand{\url}[1]{#1}
\csname url@samestyle\endcsname
\providecommand{\newblock}{\relax}
\providecommand{\bibinfo}[2]{#2}
\providecommand{\BIBentrySTDinterwordspacing}{\spaceskip=0pt\relax}
\providecommand{\BIBentryALTinterwordstretchfactor}{4}
\providecommand{\BIBentryALTinterwordspacing}{\spaceskip=\fontdimen2\font plus
\BIBentryALTinterwordstretchfactor\fontdimen3\font minus
  \fontdimen4\font\relax}
\providecommand{\BIBforeignlanguage}[2]{{%
\expandafter\ifx\csname l@#1\endcsname\relax
\typeout{** WARNING: IEEEtran.bst: No hyphenation pattern has been}%
\typeout{** loaded for the language `#1'. Using the pattern for}%
\typeout{** the default language instead.}%
\else
\language=\csname l@#1\endcsname
\fi
#2}}
\providecommand{\BIBdecl}{\relax}
\BIBdecl

\bibitem{JR:Mag_green}
T.~Chen, Y.~Yang, H.~Zhang, H.~Kim, and K.~Horneman, ``{Network Energy Saving
  Technologies for Green Wireless Access Networks},'' \emph{IEEE Wireless
  Commun.}, vol.~18, pp. 30--38, Oct. 2011.

\bibitem{CN:WIPT_fundamental}
L.~Varshney, ``{Transporting Information and Energy Simultaneously},'' in
  \emph{Proc. IEEE Intern. Sympos. on Inf. Theory}, Jul. 2008, pp. 1612 --1616.

\bibitem{CN:Shannon_meets_tesla}
P.~Grover and A.~Sahai, ``{Shannon Meets Tesla: Wireless Information and Power
  Transfer},'' in \emph{Proc. IEEE Intern. Sympos. on Inf. Theory}, 2010, pp.
  2363 --2367.

\bibitem{JR:WIPT_fullpaper}
D.~W.~K. Ng, E.~S. Lo, and R.~Schober, ``{Wireless Information and Power
  Transfer: Energy Efficiency Optimization in OFDMA Systems},''
  \url{http://arxiv.org/pdf/1303.4006v1.pdf}, Tech. Rep.

\bibitem{CN:WIP_receiver}
X.~Zhou, R.~Zhang, and C.~K. Ho, ``{Wireless Information and Power Transfer:
  Architecture Design and Rate-Energy Tradeoff},'' in \emph{Proc. IEEE Global
  Telecommun. Conf.}, Dec. 2012.

\bibitem{Report:Wire_tap}
A.~D. Wyner, ``{The Wire-Tap Channel},'' Tech. Rep., Oct 1975.

\bibitem{JR:Artifical_Noise1}
S.~Goel and R.~Negi, ``{Guaranteeing Secrecy using Artificial Noise},''
  \emph{IEEE Trans. Wireless Commun.}, vol.~7, pp. 2180 -- 2189, Jun 2008.

\bibitem{JR:Kwan_physical_layer}
D.~W.~K. Ng, E.~S. Lo, and R.~Schober, ``{Secure Resource Allocation and
  Scheduling for OFDMA Decode-and-Forward Relay Networks},'' \emph{IEEE Trans.
  Wireless Commun.}, vol.~10, pp. 3528--3540, 2011.

\bibitem{JR:SeDumi}
J.~F. Sturm, ``{Using SeDuMi 1.02, A MATLAB Toolbox for Optimization over
  Symmetric Cones},'' \emph{{Optimiz. Methods and Software}}, vol. 11-12, pp.
  625--653, Sep. 1999.

\bibitem{book:convex}
S.~Boyd and L.~Vandenberghe, \emph{{Convex Optimization}}.\hskip 1em plus 0.5em
  minus 0.4em\relax {Cambridge University Press}, 2004.

\bibitem{report:tgn}
{IEEE P802.11 Wireless LANs, ``TGn Channel Models", IEEE 802.11-03/940r4},
  Tech. Rep., May 2004.

\bibitem{report:80211af}
H.-S. Chen and W.~Gao, ``{MAC and PHY Proposal for 802.11af},'' Tech. Rep.,
  Feb., [Online]
  \url{https://mentor.ieee.org/802.11/dcn/10/11-10-0258-00-00af-mac-and-phy-proposal-for-802-11af.pdf}.

\end{thebibliography}

\end{document}